\documentclass[11pt,runningheads]{llncs}
\usepackage[vlined,longend,ruled]{algorithm2e}

\usepackage[T1]{fontenc}
\usepackage{graphicx}
\usepackage{url}

\usepackage{stmaryrd,euscript,float,amssymb,amsmath,tikz,hyperref,bookmark}
\usepackage{color}

\newtheorem{invariant}{Invariant}

\newtheorem{affirmation}{Claim}

\begin{document}
%
%\title{Contribution Title\thanks{Supported by organization x.}}
\title{On the power of standard DFS and BFS}
%\title{Draft for a TP and UIG recognition}
%
\titlerunning{On the power of standard DFS and BFS}
% If the paper title is too long for the running head, you can set
% an abbreviated paper title here
%
\author{Binh-Minh Bui-Xuan\inst{1}
\and Michel Habib\inst{2}
\and Fabien de Montgolfier\inst{2}
\and Renaud Torfs\inst{2}}
%
%\authorrunning{M. Habib}
% First names are abbreviated in the running head.
% If there are more than two authors, 'et al.' is used.
%
%\institute{Princeton University, Princeton NJ 08544, USA \and
%Springer Heidelberg, Tiergartenstr. 17, 69121 Heidelberg, Germany
%\email{lncs@springer.com}\\
%\url{http://www.springer.com/gp/computer-science/lncs} \and
%ABC Institute, Rupert-Karls-University Heidelberg, Heidelberg, Germany\\
%\email{\{abc,lncs\}@uni-heidelberg.de}}
\institute{LIP6, CNRS, Sorbonne Universit\'e, \email{buixuan@lip6.fr}
\and IRIF, Universit\'e Paris Cit\'e, \email{[habib,fm,torfs]@irif.fr}}
\maketitle              % typeset the header of the contribution

"Simplicity is a great virtue but it requires hard work to achieve it and education to appreciate it. And to make matters worse: complexity sells better."

\hfill Edsger W. Dijkstra [EWD86]
\begin{abstract}
It is well-known since the seventies of last century that Depth First Search (DFS) can be used to compute strongly connected components~\cite{Tarjan72} and Breadth First Search (BFS) can be used to compute distance in graphs~\cite{H73}.
We furthermore demonstrate that these standard graph searches are powerful enough to recognize and certify several well-structured graph classes.
Specifically, we provide a single DFS approach for recognizing and certifying trivially perfect graphs that is significantly simpler than previous methods using LexBFS~\cite{Chu08}.
We further show that a single BFS can recognize split graphs and bipartite chain graphs, and we improve upon the triple LexBFS algorithm for proper interval graphs~\cite{Corneil04} by proposing a two consecutive BFS recognition scheme.
These results are underpinned by characterizations using vertex orderings that avoid specific patterns~\cite{FeuilloleyH21}.
Finally, we provide a structural study of connected proper interval graphs, proving that their characterizations via special orderings are unique up to reversal and the permutation of true twins.

\end{abstract}

%%%%%%%%%%%%%%%%%%%%%%%%%%%%%%%%%%%%%%%%%%%%%%%%%%
%%%%%%%%%%%%%%%%%%%%%%%%%%%%%%%%%%%%%%%%%%%%%%%%%%
\section{Introduction and notations}
\label{section.intro}
We think that the potential of basic graph searches such as BFS and DFS is sometimes underestimated in particular for the recognition algorithms of well structured graph classes.
As an example, we present here how a single DFS search can be used to recognize Trivially Perfect Graphs, this approach is much simpler than the one developed by~\cite{Chu08} using a LexBFS graph search.
Similarly we show that a single BFS can recognize  split graphs, bipartite chain graph.
Furthermore, we propose two consecutive BFS for proper interval graph recognition improving the three consecutive LexBFS proposed in~\cite{Corneil04} and also improving many other published algorithms.
All algorithms are naturally certifying and can be used to recognize the complement class within the same complexity.
A key tool we use to analyze these algorithms for graph classes recognition are ordered patterns.
As a consequence of our algorithmic study, we obtain a nice structural property for proper interval graphs, showing that there is a unique ordering up to reversal and permutation of true twins vertices when the graph is connected.

The graphs considered here are simple, undirected, loopless and finite.
For an undirected graph $G=(V, E)$ with vertex set $V$ and edge set $E$, we denote for a vertex $v \in V$ by $N(v)$ its neighbourhood and $degree(v)=|N(v)|$.
Furthermore its closed neighbourhood $N(v)\cup \{v\}$ is denoted by $N[v]$.
We let $n = |V|$ and $m = |E|$.
We denote by $G[V']$ the \emph{induced subgraph} $(V',E')$ of $G=(V,E)$ on the subset $V'\subseteq V$, where for every pair $u,v \in V', uv \in E'$ if and only if $uv \in E$.
A graph class is said to be \emph{hereditary} if it is closed under induced subgraphs.
Two vertices $x,y$ are \emph{twins} if $N(x) \setminus \{y\} = N(y) \setminus \{x\}$.
If in addition, $xy \in E$, we say $x,y$ are \emph{true twins}, otherwise they're \emph{false twins}.
A \emph{module} $M \subseteq V$ is a collection of vertices that have the same neighbourhood outside of $M$: for all $x,y \in M: N(x) \setminus M = N(y) \setminus M$.
$\emptyset$, every singleton $\{x\}$ and the whole vertex set $V$ are called \emph{trivial} modules.
A \emph{prime} graph is a graph that admits only trivial modules.
\smallskip

\textbf{Degree-guided graph searching:}
A graph search with maximum (resp.\ minimum) degree priority consists in taking, when there is a \textbf{tie} with many eligible vertices, one of the vertices with largest (resp.\ smallest) degree.
The same priority rule among eligible vertices is used throughout the graph search (sort once, use forever).
This is in contrast to Maximal Cardinality Search or \texttt{MCS} search \cite{TY84}, in which the tie-break is solved by choosing a vertex that has a maximum of already visited neighbours (dynamic sort).

BFS refers to the standard Breadth First Search implemented with a queue data structure. Depth First Search (DFS) proceeds by exploring an adjacent vertex to the most recently visited vertex, backtracking when no unvisited neighbours remain and can be implemented with a simple stack.

Four searches can thus be defined:
\texttt{minBFS} and \texttt{maxBFS} are variants of the usual Breadth-First Search (BFS), where the unmarked vertices of the vertex currently under visit are put in the queue in increasing (resp.\ decreasing) degree ordering;
likewise, \texttt{minDFS} and \texttt{maxDFS} are variants of the usual Depth-First Search (DFS), where the recursive search function is called on the neighbour with smallest (resp.\ largest) degree of the vertex currently under visit.
They have linear time complexity:
in a first step, the vertex set can be bucket sorted by degree;
for every $u$ following this order, $u$ can be added to the end of the adjacency list of every $v$ belonging to the neighbourhood of $u$;
then, the second step is a standard BFS/DFS.
By \texttt{minXFS}, we denote either \texttt{minBFS} or \texttt{minDFS}.
\smallskip

\textbf{Search trees:} In a BFS-tree (resp. a DFS-tree),  the parent of each vertex is its first (resp. last) visited neighbour occurring before it.
Our DFS-trees are particular cases of $\mathfrak{L}$-trees as defined in \cite{BeisegelDKKPSS21}.
\smallskip

\textbf{Ordered patterns: } Let $\sigma$ be a total ordering on $V$. We denote $\sigma^d$ its mirror ordering, or order in reverse.
In this article, we will make extensive use of characterizations using ordered patterns of certain classes of hereditary graphs, first defined in \cite{Damaschke90}.
Using the notations defined in  \cite{FeuilloleyH21}, a graph class $\cal C$ is characterized by an ordered pattern $P$ when:

$G \in \cal C$ iff $G$ admits a total ordering of the vertices avoiding the ordered  pattern $P$.

As examples think of comparability graphs (resp. chordal graphs) and the left pattern (resp. right pattern) of Figure \ref{patternsTPG}.  Similarly a cocomparability graph (complement of a comparability graph) can be characterized by the existence of a total  ordering of its vertices
avoiding the pattern of Figure \ref{cocomp}.
This read as: $G$ is cocomparability graph iff there exists an ordering $\tau$ of the vertices such that for any triple of vertices $a <_{\tau} b <_{\tau}c$, we do not have $ac$ is an edge while $ab, bc$ are not an edges.
We denote such ordering $\tau$ as a \textit{cocomp} ordering.

\begin{figure}
\begin{center}
\begin{tikzpicture}
%        [scale=1,auto=left,every node/.style=
%       {circle,draw,fill=black!5}]
        \node (a) at (0,0) {a};
        \node (b) at (1,0) {b};
        \node (c) at (2,0) {c};
        \draw (a) to[bend left=50] (c);
        \draw[dashed] (a) to (b);
        \draw[dashed] (b) to (c);
\end{tikzpicture}

\end{center}
\caption{\label{cocomp} Cocomparability pattern.
A graph is a cocomparability graph iff there exists an ordering of its vertices avoiding the cocomparability pattern.
By drawing conventions the patterns is ordered from left to right.
To satisfy the pattern, a black (resp. dashed) edge means a mandatory edge (resp. non-edge).
Otherwise, the pattern is avoided.}
\end{figure}
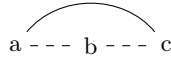

Search orderings also can be characterized using pattern. We use indeed:
\begin{proposition}\label{fourpoint}\cite{CorneilK08}
An ordering $\tau$ is a BFS (resp. DFS) ordering iff for any $a <_{\tau} b <_{\tau}c$, when $ac$ is an edge while $ab$ is not an edge, there exists $d <_{\tau} a$ (resp. $a <_{\tau} d <_{\tau} b$) such that $db$ is an edge.
\end{proposition}

\begin{remark}\label{complement}
If an algorithm based on DFS or BFS recognizes some hereditary class of graphs, we can also derive an algorithm recognizing the complement class within the same complexity. One way to be convinced is to consider that BFS and DFS can be implemented using partition refinement which is a technique that can be used to compute on a graph $G$ having the adjacency lists of $\overline{G}$.
This generalizes to \texttt{minBFS} and \texttt{maxBFS} by just moving to the dual of the degree orderings of the vertices.
\end{remark}

%%%%%%%%%%%%%%%%%%%%%%%%%%%%%%%%%%%%%%%%%%%%%%%%%%
%%%%%%%%%%%%%%%%%%%%%%%%%%%%%%%%%%%%%%%%%%%%%%%%%%
\section{First applications of forbidden patterns}\label{aperitif}
It is well-known that standard BFS can be used to recognize Bipartite graphs and Forest.
Furthermore, we can easily recognize some other graph classes.
A graph is a \textit{split graph} if there exists a partition of the vertices such that the subgraph induced by the first part is a clique, and the subgraph induced by the second part is an independent set. Note that the split graphs are connected, except (possibly) for isolated vertices.
A \textit{star} is a (split) graph where at most one node has several neighbours. 
From \cite{Damaschke90} we know that these classes can be characterized  by the patterns described in the following  Figure \ref{SplitStar}.
\begin{proposition}
Let $\tau$ be the visiting ordering of a $\texttt{maxBFS}$ on a graph $G$.
$\tau^d$ avoids the left pattern (resp. right) of Figure \ref{SplitStar} iff $G$ is a star (resp. a split graph).
\end{proposition}
\begin{proof}
Necessarily if $G$ is split graph a $\texttt{maxBFS}$ starts at a vertex of the clique and will visit all vertices of the clique before visiting the independent set. So the dual or reverse of the visiting ordering necessarily avoids the split pattern if $G$ is a split graph. The converse is obtained with the characterization of split graphs in \cite{Damaschke90}.
The proof is similar for stars. \qed
\end{proof}
\begin{figure}
\begin{center}

\begin{tabular}{cc}
\begin{tikzpicture}

%        [scale=1,auto=left,every node/.style=
%       {circle,draw,fill=black!5}]
        \node (a) at (0,0) {a};
        \node (b) at (1,0) {b};
        \node (c) at (2,0) {c};
%       \draw (a) to[bend left=50] (c);
        \draw (a) to (b);
%        \draw (b) to (c);
\end{tikzpicture}

$\qquad$
\begin{tikzpicture}
%        [scale=1,auto=left,every node/.style=
%       {circle,draw,fill=black!5}]
        \node (a) at (0,0) {a};
        \node (b) at (1,0) {b};
        \node (c) at (2,0) {c};
%        \draw[dashed] (a) to[bend left=50] (c);
        \draw (a) to (b);
        \draw [dashed] (b) to (c);
        
\end{tikzpicture}

\end{tabular}
\end{center}
\caption{\label{SplitStar} The forbidden pattern of stars (left) and split graphs (right).}
\end{figure}
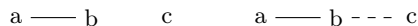

It is well-known that a similar approach with  \texttt{maxBFS} also leads to a certifying algorithm for the recognition of threshold graphs, see \cite{MP95}, which are a particular case of split graphs.

A \textit{bipartite chain graph} is a  $2K_2$-free bipartite graph, that is a bipartite graph that does not have two independent edges that are not  linked by a third edge (\emph{i.e.} no induced complement of a $C_4$), see \cite{HeggernesK07}.
Another way to understand a bipartite chain graph, is a bipartite for which, in each class, one can order the neighbourhoods by inclusion. That is, with a bipartition, $A,B$, $A=a_1, \dots,  a_{|A|}$ satisfies $N(a_1) \supseteq N(a_2), \dots, \supseteq N(a_{|A|})$ and  $B=b_1, \dots,  b_{|B|}$ satisfies $N(b_1) \subseteq N(b_2) \dots \subseteq N(b_{|B|})$.  It should be noticed that this bipartition is not necessarily unique, think of a complete graph.

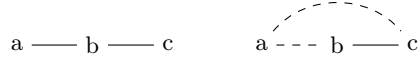
\begin{figure}
\begin{center}
\begin{tabular}{cc}
\begin{tikzpicture}
%        [scale=1,auto=left,every node/.style=
%       {circle,draw,fill=black!5}]
        \node (a) at (0,0) {a};
        \node (b) at (1,0) {b};
        \node (c) at (2,0) {c};
%        \draw[dashed] (a) to[bend left=50] (c);
        \draw (a) to (b);
        \draw  (b) to (c);
        
\end{tikzpicture}
$\qquad$
\begin{tikzpicture}

%        [scale=1,auto=left,every node/.style=
%       {circle,draw,fill=black!5}]
        \node (a) at (0,0) {a};
        \node (b) at (1,0) {b};
        \node (c) at (2,0) {c};
      \draw (a)[dashed] to[bend left=50] (c);
        \draw (a)[dashed] to (b);
       \draw (b)  to (c);
\end{tikzpicture}
\end{tabular}
\end{center}
\caption{\label{otherbipchaine} The forbidden patterns of bipartite chain graphs.}
\end{figure}

A particular subclass of bipartite chain graphs is used in graph theory for coloration problems under the name Half-Graphs, see
\cite{Erdos84} mainly because they contains long chains of inclusion of neighbourhoods.

Of course to recognize bipartite chain graph one can test its bipartiteness using a BFS and then check the inclusion of neighbourhoods. This can be done in linear time, but using our knowledge of graph searches  we will now prove that a simple $\texttt{minBFS}$ can do the job.
As proved in \cite{FeuilloleyH21} the 2 patterns of Figure \ref{otherbipchaine} characterize bipartite chain graphs.
\begin{proposition}[\cite{FeuilloleyH21}]\label{autrespatterns}
A graph $G$ is a bipartite chain graph iff it admits a total ordering of its vertices avoiding the patterns of Figure \ref{otherbipchaine}.
\end{proposition}

\begin{theorem}\label{bipchainprop}
Let $\tau$ be the visiting ordering of a $\texttt{minBFS}$ on a graph $\overline{G}$.
$\tau$ avoids the two patterns of Figure \ref{otherbipchaine}  iff $G$ is a bipartite chain graph. 
\end{theorem}
\begin{proof}
$\rightarrow$
Let us consider a bipartite chain graph $G$  with one of its bipartition, $A,B$, $A=a_1, \dots,  a_{|A|}$ satisfies $N(a_1) \supseteq N(a_2), \dots, \supseteq N(a_{|A|})$ and  $B=b_1, \dots,  b_{|B|}$ satisfies $N(b_1) \subseteq N(b_2) \dots \subseteq N(b_{|B|})$.  
As shown in Figure \ref{2-chainG} note that in $\overline{G}$ the two sides of the bipartite $A, B$  are now cliques and the total order of the neighbourhoods between the 2-sides are just reversed.

$\overline{N(a_1)} \subseteq \overline{N(a_2)}, \dots, \subseteq \overline{N(a_{|A|})}$ and   $\overline{N(b_1)} \supseteq \overline{N(b_2)} \dots \supseteq \overline{N(b_{|B|})}$. 

Let $\tau$ be a $\texttt{minBFS}$ applied on $\overline{G}$.  Suppose that $\tau$ starts in $a_1$. It means that $|\overline{N(a_1)}| \leq |\overline{N(b_{|B|})}|$ or equivalently $d_G(a_1) \geq d_Gb_{|B|})$.

Furthermore we also assume that the bipartition $A, B$ of $\overline{G}$ has $B$ maximal over all bipartitions into 2 cliques with $a_1 \in A$  and $b_{|B|} \in B$.

\begin{affirmation} \label{consecutifs}
Vertices of $A$ (resp. $B$) are consecutive in $\tau$.
\end{affirmation}

\begin{proof}
If $\overline{N(a_1)}=A-\{a_1\}$ then the $\texttt{minBFS}$ will visit first $A$.

Else if $\overline{N(a_1)}\cap B \neq \emptyset$, let  $b_i=max_j\{ b_j \in \overline{N(a_1)}\}$.
If some $a_j$ is after $b_i$ in $\tau$. It means that $|\overline{N(b_i)}| \leq |\overline{N(a_j)}|$. But $|\overline{N(b_i)}|= |A| +|B|-1$ which is maximum possible and therefore there is a complete bipartite from $a_j$ to $a_{|A|}$ to $B$. And so we can change the bipartition moving the vertices
$a_j, \dots , a_{|A|}$ from $A$ to $B$.
But this is impossible since it contradicts our assumption on the bipartition with $B$ maximal.
Doing so $A$ and then the search moves to $B$.
So going back to $G$ the search visits separately the 2 sides of the bipartite graph and the left pattern of Figure \ref{otherbipchaine} is clearly avoided in  $\tau$. Hence the claim is obtained.
\end{proof}

Using Claim \ref{consecutifs} $\tau$ avoids the second pattern of Figure 
\ref{otherbipchaine}.

Let us now consider the second pattern of Figure \ref{otherbipchaine}, which in fact means that $\tau$ must be a simplicial elimination scheme in $\overline{G}$.

\begin{affirmation} \label{simpli}
$\tau$ is a simplicial elimination scheme on $\overline{G}$.
\end{affirmation}
\begin{proof}

Using Claim \ref{consecutifs} we know that $\tau$ visits first $A$ then $B$.

$\tau$ starts with a vertex of minimum degree in $\overline{G}$ namely $a_1 \in A$

If $\overline{N(a_1)}=A-\{a_1\}$ then 
$a_1$ is clearly simplicial in $\overline{G}$ because  the only neighbours of $a_1$ are in the clique $A$.

Else if $\overline{N(a_1)}\cap B \neq \emptyset$, using the inclusion of the neighbourhoods, $a_1$ is clearly simplicial in $\overline{G}$ this property remains true
for all other vertices of $A$.
Since $B$ is also a clique of $\overline{G}$, the visiting of $B$ is also a simplicial elimination scheme.
So $\tau$ is clearly a simplicial elimination scheme of $\overline{G}$.

\end{proof}

So using these two Claims $\tau$ avoids the two forbidden patterns of bipartite chain graphs, certifying that $G$ is a bipartite chain graph.

$\leftarrow$ Conversely  if a graph $G$ admits a total ordering of its vertices avoiding the patterns of Figure \ref{otherbipchaine}, then using the pattern  characterization of bipartite chain graphs of \cite{FeuilloleyH21} and Proposition \ref{autrespatterns} we know that $G$ is a bipartite chain graph.
\end{proof}
\begin{figure}
\begin{center}

\begin{tabular}{cc}
\begin{tikzpicture}

%        [scale=1,auto=left,every node/.style=
%       {circle,draw,fill=black!5}]
        \node (1) at (0,0) {1};
        \node (2) at (1,0) {2};
        \node (3) at (2,0) {3};
        
        \node (a) at (0,-2) {a};
        \node (b) at (1,-2) {b};
        \node (c) at (2,-2) {c};
    \draw (a) to (1);
     \draw (a) to (2);
      \draw (a) to (3);
      \draw (b) to (2);
      \draw (b) to (3);
         \draw (c) to (3);  
 \node (G) at (1, -3) {G};
\end{tikzpicture}

$\qquad$
\begin{tikzpicture}
%        [scale=1,auto=left,every node/.style=
%       {circle,draw,fill=black!5}]
        \node (1) at (0,0) {1};
        \node (2) at (1,0) {2};
        \node (3) at (2,0) {3};

        \node (a) at (0,-2) {a};
        \node (b) at (1,-2) {b};
        \node (c) at (2,-2) {c};
       \draw (1) to[bend  left=50] (3);
       \draw (a) to[bend left=50] (c);
       \draw (1) to (2);
       \draw (2) to (3);
        \draw (a) to (b);
        \draw (b) to (1);
        \draw (c) to (1);
        \draw (c) to (2);
        \draw  (b) to (c);
\node (G) at (1, -3) {$\overline{G}$};        
\end{tikzpicture}

\end{tabular}
\end{center}
\caption{\label{2-chainG} Example of bipartite chain graphs. And $\tau= a,b,c,1, 2, 3$ is a $\texttt{minBFS}$ visiting ordering of $\overline{G}$, avoiding the patterns of Figure \ref{otherbipchaine}.}
\end{figure}
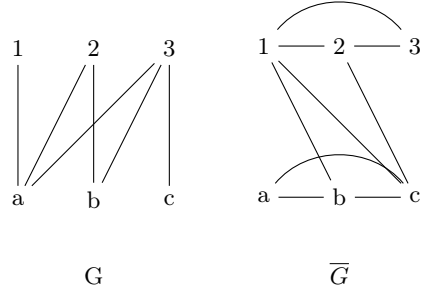

Using Remark \ref{complement} we can compute in linear time a $\texttt{minBFS}$ on $\overline{G}$ from $G$. To certify that an ordering of the vertices avoids the patterns of Figure \ref{otherbipchaine} can be done in linear time as explained in \cite{DFHP21}.
This result will be  generalized in Section \ref{section.uig}.

%%%%%%%%%%%%%%%%%%%%%%%%%%%%%%%%%%%%%%%%%%%%%%%%%%
%%%%%%%%%%%%%%%%%%%%%%%%%%%%%%%%%%%%%%%%%%%%%%%%%%
\section{Trivially perfect graphs}
Originally introduced as comparability graphs of trees by~\cite{Wolk62}, this class was later named trivially perfect graphs~\cite{Golumbic78}.
They are also known as quasi-threshold graphs~\cite{JINGHO1996}.
These graphs possess interesting properties suitable for developing new applications, as shown in~\cite{PaulP24}.
The next theorem presents their most recognized characterizations.

\begin{theorem}~\cite{Golumbic2nd,Wolk62,JINGHO1996,FeuilloleyH21}
 A graph is a trivially perfect graph iff it satisfies one of the following equivalent conditions:
\begin{enumerate}
\item a graph in which, for every induced subgraph, the size of a maximum independent set is equal to the number of maximal cliques.
\item \label{item:quasi-threshold} a graph that can be constructed recursively the following way: creation of a single node graph; disjoint union of two graphs; addition of one universal vertex.
\item \label{item:p4-c4} a $(C_4,P_4)$-free graph.
\item \label{item:comparability-of-tree} a comparability graph of a rooted tree, that is the comparability graph of a partial order in which for every element $x$, the elements of $\{y : y<x\}$ can be linearly ordered.
\item the intersection graph of a set of nested intervals (that is of intervals such that for every two intersecting intervals, one is included in the other).
\item a cograph that is also an interval graph.
\item a cograph that is also a chordal graph.
\item a graph that admits an elimination ordering $x_1, x_2, \dots, x_n$ such that for every $i$,
$x_i$ is universal to its connected component in $G[\{x_{i+1}, x_{i+2}, \dots, x_n\}]$.
\item
A graph that admits a vertex ordering avoiding the 2 patterns in Figure \ref{patternsTPG}.
\end{enumerate}
\end{theorem}
\begin{figure}[!h]
\begin{center}
\begin{tabular}{cc}
\begin{tikzpicture}
%        [scale=1,auto=left,every node/.style=
%       {circle,draw,fill=black!5}]
        \node (a) at (0,0) {a};
        \node (b) at (1,0) {b};
        \node (c) at (2,0) {c};
        \draw[dashed] (a) to[bend left=50] (c);
        \draw (a) to (b);
        \draw (b) to (c);
\end{tikzpicture}
$\qquad$
\begin{tikzpicture}
%        [scale=1,auto=left,every node/.style=
%       {circle,draw,fill=black!5}]
        \node (a) at (0,0) {a};
        \node (b) at (1,0) {b};
        \node (c) at (2,0) {c};
        \draw (a) to[bend left=50] (c);
        \draw[dashed] (a) to (b);
        \draw (b) to (c);
\end{tikzpicture}
\end{tabular}
\end{center}
\caption{\label{patternsTPG} The two forbidden patterns of a trivially perfect ordering.}
\end{figure}
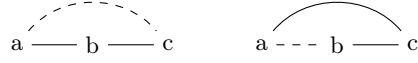

Characterization n°3 provides a certificate of exclusion from the class of trivially perfect graphs.
The following lemmas are baselines for our recognition algorithm.
\begin{lemma}
	\label{nottwin}
	Let $G = (V,E)$ be a graph.
	If there exist $x,y,z \in V$ such that $xy\in E$, $yz \in E$, $xz\notin E$ and $degree(y) \leq degree(z)$, then $G$ has an induced $C_4$ or an induced $P_4$.
\end{lemma}
\begin{proof}
        Since $y$ has lower degree, the existence of a private neighbor $x$ of $y$ implies the existence of a private neighbor of $z$, a vertex $t\in N(z)\setminus N(y)$.
	Then, either $xt \in E$ and $x,y,z,t$ is an induced $C_4$,
	or $xt \notin E$ and $x,y,z,t$ is and induced $P_4$ of $G$.
\end{proof}
\begin{lemma}\label{dfscondition}
	Let $G = (V,E)$ be a graph and let $T_G$ be the search tree of a \emph{\texttt{maxDFS}} search on $G$.
	If there exists a node $x$ of $T_G$ whose parent is $y$ with an ancestor $z$ such that $xz \notin E$ then $G$ has an induced $C_4$ or $P_4$ and hence not a trivially perfect graph.
	Otherwise, $G$ is trivially perfect and is equal to the underlying undirected graph corresponding to the transitive closure of $T_G$.
\end{lemma}
\begin{proof}
	If $G$ is not connected then \texttt{maxDFS} visit components one by one. Assume then that $G$ is connected.
        
        If such vertex $x$ exists, we assume w.l.o.g.\ that it is chosen maximally close to the root of $T_G$.
        Hence, $yz\in E$.
	Since $T_G$ is a DFS search tree, we also have $xy \in E$.
	Since $z$ is an ancestor of $y$ (i.e. the \texttt{maxDFS} search visited $z$ first) we have that $degree(y) \leq degree(z)$.
        Combining this to $x \in N(y)\setminus N(z)$, we deduce from
	Lemma \ref{nottwin} that $G$ has an induced $C_4$ or $P_4$.
	
        Otherwise, every vertex of $T_G$ share an edge with all its ancestors, meaning the transitive closure of $T_G$ is a partial subgraph of $G$.
        To prove equality, let now be any edge $uv \in E$ of $G$.
        Since $T_G$ is a DFS search tree, we assume w.l.o.g.\ that $u$ is an ancestor of $v$ in $T_G$.
        In other words, $uv$ is an edge of the transitive closure of $T_G$. $G$ is thus trivially perfect.
\end{proof}

This lemma alone yields a recognition algorithm for trivially perfect graphs: perform a \texttt{maxDFS} and check if such a vertex $x$ exists. A naive implementation would yield a runtime in $O((n+m)\times diameter)$.
Both step can be merged as shown by Algorithm~\ref{TPG} below. For building a positive certificate easy to check, we now examine the \texttt{maxDFS} ordering properties.

\begin{lemma}\label{dfsordercondition}
	Let $G = (V,E)$ be a graph, $T$ a search tree of a \emph{\texttt{maxDFS}} search of $G$, $\sigma$ the ordering of search order of $T$, $x,y$ and $z$ vertices of $G$ such that $xy \in E$, $x$ is before $y$ in $\sigma$, and $z \in N(y) \setminus N(x)$.
	Then, we can instantly output an induced $C_4$/$P_4$ of $G$.
\end{lemma}

\begin{proof}
	First suppose there is no vertex $b$ before $x$ in $\sigma$ such that $xb \in E$. 
	Then $x$ is first vertex of its connected component to be seen. Since it is a \texttt{maxDFS} ordering then $x$ has maximum degree in that component, implying $deg(x) \geq deg(y)$.
	Since $z \in N(y) \setminus N(x)$ then for $x$ to have a greater or equal degree than $y$ there must exists a vertex $v \in N(x) \setminus N(y)$.
	
	Suppose now that $x$ has a neighbor before it. Let
	$b$  be the last vertex before $x$ in $\sigma$ such that $bx \in E$.
	If $by\notin E$ take $v=b$ and we are done.
	So let us suppose  $by \in E$. $b$ is the parent of $x$ in \texttt{maxDFS} since it is the last vertex before $x$. When $x$ is visited, $b$ is the active vertex and $y$ is eligible but not chosen. The degree tie-breaking rule implies  $degree(x) \geq degree(y)$.
	Since $z \in N(y) \setminus N(x)$ then there must exists a vertex $v \in N(x) \setminus N(y)$.
	Then, $\{x,y,z,v\}$ induces either a $P_4$ or a $C_4$.

In both cases, either edge $zv$ exists or not, $\{x,y,z,v\}$ induces a $P_4$ or a $C_4$.
\end{proof}

\begin{proposition}\label{vertex ordering}
Let $\sigma$ be the visiting ordering of a \texttt{maxDFS} on a graph $G$.
$G$  is a trivially perfect graph iff  $\sigma$ forbids both 2 patterns of Figure \ref{patternsTPG}.
\end{proposition}
\begin{proof}
$\leftarrow$
follows from the above  trivially perfect graphs' characterizations.

$\rightarrow$ Let $\sigma$ be the visiting ordering of maxDFS on a graph G.
Suppose that $\sigma$ admits the first pattern of Figure \ref{patternsTPG},  then using Lemma~\ref{dfsordercondition} with $x=a, y=b$ and $z=c$ then $G$ contains a $C_4$ or a $P_4$ and is not trivially perfect.
Suppose that $\sigma$ admits the second pattern of Figure \ref{patternsTPG},  then using Lemma~\ref{dfsordercondition} with $x=a, y=c$ and $z=b$ then $G$ contains a $C_4$ or a $P_4$ and is not trivially perfect.
\end{proof}

This yields therefore to a second recognition algorithm for trivially perfect graphs: perform a \texttt{maxDFS} and check (in linear time \cite{DFHP21}) if $\sigma$ does not contain one of the two forbidden patterns.
But in fact one just have to check one of the two pattern:
\begin{proposition}\label{onepattern}
A DSF-ordering avoiding the first pattern of Figure \ref{patternsTPG} avoids also the other one.
\end{proposition}
\begin{proof}
Let us consider a DFS-ordering $\sigma$ avoiding the first pattern of Figure \ref{patternsTPG} and suppose there exists an occurrence of the second pattern. 3 vertices $x <_{\sigma} y <_{\sigma}$ with $xz, yz \in E(G)$ and $xy \notin E(G)$. Since $\sigma$ is a DFS-ordering necessarily there is a DFS-path $\mu=[x_0=x, x_1, \dots, x_k=y]$ in $G$. Let $j<k$ be the index maximum such that $xx_j \in E(G)$. It is well-defined since $xx_1 \in E(G)$ and $xx_k=xy \notin E(G)$.
But then $x, x_j, x_{j+1}$ is a pattern of the first type, a contradiction.
\end{proof}
\begin{theorem}
	There is a linear time certifying recognition algorithm for trivially perfect graphs using \emph{\texttt{maxDFS}} once and only once.
\end{theorem}

\begin{algorithm}
\caption{Trivially perfect graph Recognition}\label{TPG}

\KwData{a connected graph $G$, given by its adjacency lists}
\KwResult{Either "$G$ is not a Trivially Perfect Graph"; or a total ordering of the vertices \texttt{rank} and a tree \texttt{parent};}

\textbf{Initializations:}\;

sort the vertices  in non increasing degrees $\tau$: $x_1, \dots x_n$\;

sort the adjacency lists with respect to $\tau$\;

$countin\leftarrow 1$; $countout \leftarrow 1$\;
 Note : countin and countout are \textbf{global} variables\;
\ForAll{$v \in V(G) $}{  $rank(v) \leftarrow 0$;  $Visited(v) \leftarrow FALSE$\;} 

$DFGS(G,x_1)$\;

 \end{algorithm}

\begin{algorithm}
\caption{$DFS(G, u)$}\label{Recursive DFS}
$Ancestors(u) \leftarrow 0$\;

\ForAll{$v \in N(u)$ in decreasing degree ordering}{
  \uIf{$rank(v)\neq 0$ and $Visited(v)=false$ (Remark: $v$ being visited)}
    { 
     $Ancestors(u) \leftarrow Ancestors(u) +1$
    }
    \uElseIf{$rank(v)= 0$ (Remark: $v$ unvisited) }
        {
          $rank(v) \leftarrow countin$\;
          $countin \leftarrow countin +1$\;
          $parent(v)\gets u$\;
          $DFS(G, v)$\;
          $Visited(v) \leftarrow true$\;
          $countout \leftarrow countout +1$\;
        }
    \Else{\emph{(Remark: $v$ visited.)} No nothing}
}
\If{$countin - countout  \neq Ancestors(u) $}{
         G is not a Trivially Perfect Graph\;
         Output NegCertificate(G,u,parent(u)) and STOP\;
}
\end{algorithm}

\begin{algorithm}
\caption{$NegCertificate(G, x, y)$: outputs a  $C_4$ or a $P_4$ containing $x$ and $y$}\label{negCertif}
Remark: find $z$ like in Lemma \ref{dfscondition} : non-neighbor of $x$ whose visit is in progress\;
\ForAll{vertex $u\in N(x)$}{
  mark $u$ in green}
\ForAll{vertex $u\in V(G)$}{
\If{$u$ not marked in green and $rank(u)\neq 0$ and $Visited(u)=false$ }
{$z=u; break$\;}}
Remark: now find $t\in N(z)\setminus N(y)$ like in Lemma \ref{nottwin}\;
\ForAll{vertex $u\in N(y)$}{mark $u$ in yellow\;}
\ForAll{vertex $t\in N(z)$}{
   \If{$t$ not marked in yellow}{
    \Return{$\{x,y,z,t\}$ (a $C_4$ or a $P_4$)}}}
\end{algorithm}

We can however, by adding three counter to each vertex, write an even simpler certifying recognition algorithm in the subsequent Theorem~\ref{theorem_tp}, thus avoiding the verification step over $\sigma$ of Ref.~\cite{DFHP21}.
The following properties are well know from DFS and can be easily proven:
\begin{property}
When $DFS(G,u)$ terminates, $countin-countout$ is the number of vertices that belongs to the path from $u$ to the root of the DFS (its first vertex), namely its ancestors.
\end{property}
\begin{property}
When $DFS(G,u)$ terminates, $Ancestor(u)$ counts the neighbours of $u$ which are its ancestors. 
\end{property}
\begin{corollary}\label{corrmich}
When $DFS(G,u)$ terminates,  $u$ is adjacent to all its ancestors iff $countin-countout = Ancestor(u)$.
\end{corollary}

\begin{theorem}
\label{theorem_tp}
If $G$ is a trivially perfect graph, Algorithm \ref{TPG} produces in linear time an ordering \texttt{rank} avoiding the 2 patterns and a tree whose transitive close is equal to $G$. If $G$ is not trivially perfect it outputs a $C_4$ or a $P_4$ in linear time.
\end{theorem}

\begin{proof}
This algorithm simply is a \texttt{maxDFS} producing a tree (field \texttt{parent}) and an ordering (field \texttt{rank}). Either the tree or the ordering may be used as positive certificate, according to Lemma~\ref{dfscondition} (for the tree) or to Proposition~\ref{vertex ordering} (for the vertex ordering).

According to Corollary \ref{corrmich} if the test at end of Algorithm \ref{Recursive DFS} fails then the condition of Lemma~\ref{dfscondition} is violated with $x=u$, $y=parent(u)$. We then just have to find the ancestor $z$ that is not a neighbor of $u$ and then the vertex $t$ to complete the $C_4$ or the $P_4$ like in Lemma \ref{nottwin} and \ref{dfscondition}. This can be done in linear time using Algorithm \ref{negCertif}.
\end{proof}

%%%%%%%%%%%%%%%%%%%%%%%%%%%%%%%%%%%%%%%%%%%%%%%%%%
%%%%%%%%%%%%%%%%%%%%%%%%%%%%%%%%%%%%%%%%%%%%%%%%%%
\section{Unit Interval Graphs (UIG), also known as Proper Interval Graphs}\label{section.uig}

An interval graph is an undirected graph formed from a set of intervals on the real line, with a vertex for each interval and an edge between vertices whose intervals intersect. 

Let $V$ be a finite set and, for any $x\in V$, let $Interval(x)\subset \mathbb{R}$ be an interval of the real line.
We call the set $R_V=\{Interval(x)\}_x$ a \textit{realizer}. Let $E(R_V)$ be the \textit{intersection} relation of the realizer : $xy\in E(R_V)$ iff   $Interval(x)\cap  Interval(y)$ is not empty. A graph $G=(V,E)$ is an \textit{Interval Graph} (resp. an \textit{unit interval graph UIG for short}) if it is the intersection graph of a collection of  intervals (resp. of unit length intervals); i.e., if there exists a realizer $R_V$ such that $E=E(R_V)$. 
Given a realizer of an interval graphs going from left to right on the real line, we can associate to this realizer a \textit{left (resp. ordering)} the total ordering of the vertices yielded by the left (resp. right) bounds of the intervals.

\subsection{Interval graphs}

It is well-known \cite{Damaschke90} that interval graphs are characterized by the 2 patterns of Figure \ref{intervals}. To be convinced think of the right ordering associated to a realizer. It should be noticed than a total ordering $\tau$ avoiding the right  pattern of Figure \ref{intervals} traverses the connected component of the graph strictly successively.
Let us define a particular existential pattern $\Phi_p$.
\begin{definition}
A total ordering $\tau$ of the vertices of a graph $G$ satisfies $\Phi_p$ if
For all $a,b,c$ such that $a <_{\tau} b <_{\tau} c$ with  $ac \in E$ and $ab \notin E$ then $\Phi(G,\tau)$ is true, where $\Phi_p$ is some logical predicate on $G, \tau$.
\end{definition}
\begin{proposition}\label{etonnante}
Let $\tau$ be an ordering of a graph $G$ avoiding the 2 patterns of Figure \ref{intervals}, then $\tau$ satisfies  $\Phi_p$ for every logical predicate $\Phi_p$.
\end{proposition}
\begin{proof}
If there exists $a,b,c$ such that $a <_{\tau} b <_{\tau} c$ with  $ac \in E$ and $ab \notin E$ it contradicts the 2 patterns of Figure \ref{intervals} since either $bc \in E$ (left pattern) or $bc \notin E$ (right pattern).
So a total ordering avoiding the 2 patterns of Figure \ref{intervals} satisfies the existential pattern  $\Phi_p$ for any predicate $\Phi_p$.
\end{proof}
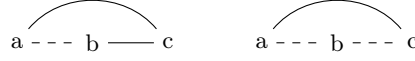
\begin{figure}
\begin{center}
\begin{tabular}{cc}
\begin{tikzpicture}
%        [scale=1,auto=left,every node/.style=
%       {circle,draw,fill=black!5}]
        \node (a) at (0,0) {a};
        \node (b) at (1,0) {b};
        \node (c) at (2,0) {c};
        \draw (a) to[bend left=50] (c);
        \draw (a)[dashed] to (b);
        \draw (b) to (c);
\end{tikzpicture}
$\qquad$
\begin{tikzpicture}
%        [scale=1,auto=left,every node/.style=
%       {circle,draw,fill=black!5}]
        \node (a) at (0,0) {a};
        \node (b) at (1,0) {b};
        \node (c) at (2,0) {c};
        \draw (a) to[bend left=50] (c);
        \draw[dashed] (a) to (b);
        \draw [dashed] (b) to (c);
\end{tikzpicture}
\end{tabular}
\end{center}
\caption{\label{intervals} The two forbidden patterns of an interval graph.}
\end{figure}

As a consequence using the 4-points conditions introduced in \cite{CorneilK08}, a total ordering of the vertices avoiding the patterns of Figure \ref{intervals}  is a particular generic search but also a DFS, a BFS or a MNS (Maximal neighbourhood search).
And also a particular LexBFS or LexDFS.
Since it is enough to produce such an ordering for the recognition of interval graphs Proposition \ref{etonnante} explains the many published algorithms for the recognition of interval graphs.   
Indeed, starting from the linear time algorithms that recognize interval graphs \cite{BoothL76} and \cite{KorteM86} using $PQ$-trees or a variation $PQR$-trees to handle the maximal cliques,  there is a long series of increasingly simple  algorithms for (resp. proper)  interval recognition. For interval graph recognition we can quote the algorithms using a series of LexBFS 6 for \cite{DCorneilOS09} and 4 for \cite{LiW14}.

It already exists many nice algorithms to recognize proper interval graphs \cite{CorneilKNOS95}, \cite{DEFIGUEIREDO1995179} using 3-sweep   \texttt{LexBFS} \cite{Corneil04}, 3-sweep \texttt{LexBFS} but with a certifying step \cite{HellH04}, and fully dynamic \cite{HellSS01}.
It is well-known from \cite{Roberts69} that in the finite case, the class of Proper Interval Graphs, namely  graphs which admits an interval representation in which no interval is included in another one, is exactly the class of Unit Interval Graphs.
Among the many characterizations of UIG, we will use the following ones.
\begin{theorem}\label{basic} \cite{Roberts69}, \cite{W1967},\cite{Fishburn85}
Let $G$ be a finite graph.  The following statements are equivalent and characterize proper interval graphs:

\begin{description}
%\item[(0)] $G$ is 
\item[(i)] $G$ is a \emph{Proper Interval Graph}, \emph{i.e.}  the intersection graphs of a set of intervals on a line, where no interval is included in another.
\item[(ii)] $G$ a \emph{Unit Interval Graph}, that is $G$ admits a realizer  in which  all the intervals have the same length. 
\item[(iii)] $G$ is  an \emph{Indifference Graph}, that is a graph 
where every node $v$ can be given a real number $k_v$ such that $(u,v)\in E$ if and only if $|k_u-k_v|\leq 1$.
\item [(iii)] $G$ is an interval graph and a $K_{1,3}$-free graph.
\item[(iv)] $G$ is a chordal graph and a $(K_{1,3}, 3$-$sun,net)$-free graph, see Figure \ref{3graphes}. 
\end{description}

\end{theorem}
\begin{figure}[h]

\begin{center}
\scalebox{.8}{
  \begin{tikzpicture}[scale=0.3]

\node(A) at (0,0) [shape=circle,draw] {$a$};
\node(a) at (-4,0) [shape=circle,draw] {$a'$};
\node(B) at  (4,0) [shape=circle,draw] {$b$};
%\node(B1) at  (4,3) [shape=circle,draw] {$b'$};
\node(C) at (8,0) [shape=circle,draw] {$c$};
\node(b) at (-4,-3) [shape=circle,draw] {$b'$};
\node(D) at (2,-3) [shape=circle,draw] {$d$};
%\node(D1) at (-1,-3) [shape=circle,draw] {$d'$};
\node(E) at (6,-3) [shape=circle,draw] {$e$};
%\node(E1) at (9,-3) [shape=circle,draw] {$e'$};
\node(F) at (4,-6) [shape=circle,draw] {$f$};

\node(c) at (-6,-6) [shape=circle,draw] {$c'$};
\node(d) at (-2,-6) [shape=circle,draw] {$d'$};

\node(e) at (12,0) [shape=circle,draw] {$g$};
\node(f) at (12,-3) [shape=circle,draw] {$h$};
\node(g) at (10,-6) [shape=circle,draw] {$i$};
\node(h) at (14,-6) [shape=circle,draw] {$j$};
\node(i) at (8,-9) [shape=circle,draw] {$k$};
\node(j) at (16,-9) [shape=circle,draw] {$l$};

\draw(A)--(B)--(C);
\draw(A)--(D)--(E);
\draw(D)--(B)--(E)--(C);
\draw(D)--(F)--(E);
\draw (a)--(b);
\draw (c)--(b)--(d);
\draw (e)--(f);
\draw (g)--(f)--(h)--(g);
\draw (i)--(g);
\draw (j)--(h);

%\node  at (4,-1.5) [scale=1] {a $P_4$};
\end{tikzpicture}
}
\end{center}
\caption[An]{\label{3graphes}$K_{1,3}$ (also known as Claw), 3$-$Sun and Net}
\end{figure}
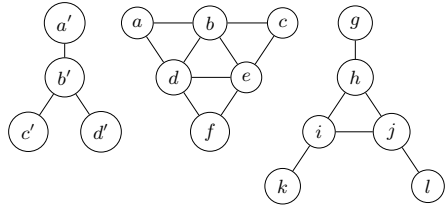

Using this above Theorem we can derive that complement of bipartite chain graphs are Proper interval graphs. First noticing that they are interval graphs as the intersection of $C_4$-free graphs and co-comparability graphs and then noticing that they are $K_{1,3}$-free graphs. So in  section \ref{aperitif}, we already process a particular case of Proper Interval graphs.
Among the other characterizations of UIG, some can be formalized in terms of  total orderings of the vertices. Let us first define these orderings:
  \begin{itemize}
  \item The \textbf{Interval Ordering} of a realizer consists in sorting the intervals by leftmost bound (or, equivalently, rightmost, since we deal with unit intervals).
  \item
A \textbf{ Consecutive Neighbourhood Ordering} of a graph is an ordering such that, for any vertex $v$, the closed neighborhood $N[v]$ of $v$ occurs consecutively in that ordering.
  \item
A  \textbf{Left-Right Ordering} of a graph is an ordering such that, for any triple $(a,b,c)$ such that $a<b<c$, if $ac\in E$ then $ab\in E$ and $bc\in E$.
\item
 An  \textbf{Indifference Ordering} of a graph is an ordering such that, for any triple $(a,b,c)$ forming a $P3$ with center $b$, i.e. $ab\in E$ and $bc\in E$ and $ac\notin E$, we have either $a<b<c$ or $c<b<a$
  \item
A  \textbf{Simplicial Elimination Ordering} is an ordering such that such that, for any vertex $v$, the neighbors of $V$ that precede it in that ordering form a clique.
\item
A  \textbf{Bisimplicial Elimination Ordering} is such that both this ordering and its mirror (\emph{i.e} its reverse ordering) are simplicial elimination orderings.
\item
A \textbf{\texttt{minBFS} Ordering} starting at $x$ of a graph $G$, denoted $\texttt{minBFS}(x)$, is the visit ordering of \texttt{minBFS} starting at vertex $x$ ($x$ comes first in that ordering).
  \end{itemize}

We already defined Left (resp. Right) ordering for interval graphs in the beginning of Section \ref{section.uig}. But in the case of Unit Interval graphs both Left and Right orderings satisfy this condition on triples. This is why we denote them as Left-Right orderings.
Following \cite{DerekBFS95}, we say for an interval graph that a vertex $x \in V(G)$ is an \textit{anchor} of a UIG $G$ if there exists a realizer  of $G$ such that $Interval(x)$ is  the leftmost interval of that realizer.
Notice that UIGs with at least two vertices admit at most two anchors, since a realizer may be \emph{mirrored}.
\begin{proposition}\cite{DerekBFS95}\label{propAnchor}
A  minimum degree vertex of the last layer of any BFS  is an anchor.
\end{proposition}
\begin{theorem}[$\star$]\cite{LO1993},\cite{HabibL09}\label{caracordres}
Let $G$ be a graph.  The following statements are equivalent:
  \begin{description}
  \item[(0)] $G$ is a UIG
\item[(i)] $G$ admits a Consecutive Neighborhood Ordering.
\item[(ii)] $G$ admits a Left-Right Ordering

\item[(iii)] $G$ admits an Indifference Ordering
\item[(iv)] $G$ admits a Bisimplicial Elimination Ordering  
  \end{description}
\end{theorem}
\begin{proof}
(0) $\rightarrow$ (i)
Let us take an interval representation of $G$. The ordering of the vertices made up the left boarder of the intervals is trivially a Consecutive Neighborhood Ordering.

(i) $\rightarrow$ (ii) Suppose that $G$ admits a consecutive Neighbourhood ordering $\tau$ of its vertices, then the forbidden pattern of Left-Right Ordering occur, since $b$ occurs inside the closed neighborhood of $a$, so this ordering also is a Left-Right Ordering.

(ii) $\rightarrow$ (iii)
A Left-Right Ordering satisfies trivially the $P_3$ condition of an indifference  ordering.

(iii) $\rightarrow$ (iv) 
It should be noticed that in terms of forbidden patterns condition (iii) states that $G$ admits an ordering of its vertices avoiding the  3 patterns of  Figure~\ref{3patterns}.
While the other  condition (iv) can be restated as follows:
$G$ admits an ordering of its vertices which is simplicial in both directions. This exactly means that this ordering 
avoiding the 2 first patterns of Figure \ref{3patterns}.
Since any ordering that avoids 3 patterns also avoids 2 of them,
therefore (iii) implies (iv).

(iv) $\rightarrow$ (0) Let us use a proof already in \cite{HabibL09}.
One pattern shows the chordality of $G$. Using Theorem \ref{basic}, it suffices to verify that none of the $K_{1,3}$, 3-sun or net admits a bisimplicial ordering. Since a characterization with ordering avoiding patterns is hereditary.
\qed
\end{proof}
\begin{figure}
\begin{center}
\begin{tikzpicture}
\node  at (1, -1) {chordal};
\node  at (4,-1)  {co-chordal};
\node  at (7, -1)  {cocomparability};
%	[scale=1,auto=left,every node/.style=		
%	{circle,draw,fill=black!5}]
	\node (a) at (0,0) {a};
	\node (b) at (1,0) {b};
	\node (c) at (2,0) {c};
	\draw (a) to[bend left=50] (c);
	\draw[dashed] (a) to (b);
	\draw (b) to (c);
	
\node (d) at (3,0) {d};
			\node (e) at (4,0) {e};
			\node (f) at (5,0) {f};
			\draw(d) to[bend left=50] (f);
			\draw (d) to (e);
			\draw[dashed] (e) to (f);

	\node (g) at (6,0) {g};
\node (h) at (7,0) {h};
\node (i) at (8,0) {i};
			\draw(g) to[bend left=50] (i);
			\draw [dashed](g) to (h);
			\draw[dashed] (h) to (i);
\end{tikzpicture}
\caption[An]{\label{3patterns}3 patterns: chordal, co-chordal and cocomparability}
\end{center}
\end{figure}
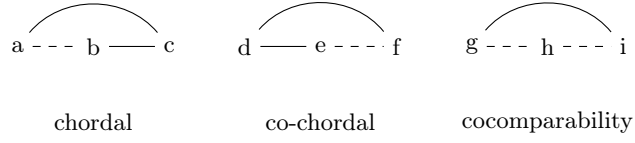

In terms of the existence of a vertex ordering avoiding some patterns, conditions (i) and (ii) correspond to the 3 pattern of Figure \ref{3patterns}.
On the other hand condition (iii) and (iv) correspond to the 2 first patterns namely chordal and co-chordal, yielding a linear time certification algorithm, see \cite{FeuilloleyH21}.
Another consequence, within this ordering, the neighbourhood of every vertex can be partitioned into at most 2 cliques, one before it and one after it.
Vertex orderings avoiding the 3 patterns of Figure \ref{3patterns} are not the same as those avoiding the 2 first patterns namely chordal and
co-chordal. To be convinced just consider the graph on 3 vertices $\{g,  h, i\}$, with edge set $gi$ and $h$ is an isolated vertex, in fact the third pattern of Figure \ref{3patterns}.
$\tau= g, h, i$ is a Bisimplicial Elimination Ordering but not a Left Ordering one.
\begin{corollary}\label{cestdejacela}
Let $\tau$ be a Bimplicial Elimination Ordering (resp.an Indifference Ordering) of an IUG $G$. If furthermore $\tau$
is a cocomp ordering then $\tau$ is a Left-Right Ordering of $G$.
\end{corollary}

We will discuss further in section \ref{unicity} the relationships of these vertex orderings for IUG.

\subsection{Recognition Algorithm}

While \cite{Corneil04} uses three consecutive LexBFS to recognize UIG, we show that two  standard BFS, are enough, as can be seen in the following very simple algorithm.

\begin{algorithm}[H]
    \KwData{A  connected graph $G$}
    \KwResult{A Left-Right Ordering iff $G$ is an IUG.}
    
sort the vertices  by non decreasing degrees $\tau$: $x_1, \dots x_n$;
\\

sort the adjacency lists with respect to $\tau$\;

Perform a BFS (therefore a \texttt{MaxBFS}) of $G$. Let $r$ be the last visited vertex\;

sort the vertices  by non increasing degrees $\theta=\tau^d$ : $x_n, \dots x_1$\;

sort the adjacency lists with respect to $\theta$\; 
Remark : This second sorting can be avoided by using circular doubled linked lists for adjacency.\;

Perform a BFS (therefore a \texttt{minBFS}) of $G$ starting from $r$, yielding a vertex ordering $\sigma$\;

\Return  $\sigma$ \;
    
    \caption{UIG recognition.}
    \label{alg:IUG}
\end{algorithm}

\begin{definition}
We say that during Pass 2 \texttt{minXFS} (recall this means  either a \texttt{minBFS} or a \texttt{minDFS})  visits vertices step by step : $\sigma(i)$ is the vertex visited at step $i$. The vertex with numbered at most $i$ are \textbf{visited}, and their unvisited neighbors  form \textbf{eligible} vertices $Eli(i)$, while the remaining vertices are \textbf{ineligible}. So we have $u\in Eli(i)$ when $\exists j>i$   $u=\sigma(j)$ and $\exists k\le i$ $\sigma(k)=v$ such that $v\in N(u)$.
\end{definition}

\begin{invariant}\label{maininv}
$Eli(i)\cup \{\sigma(i)\}$ is a clique
\end{invariant}

Algorithm~\ref{alg:IUG} is a two-pass BFS recognition algorithm, where is Pass~2 \texttt{minXFS} is implemented as \texttt{minBFS}, but using \texttt{minDFS} instead would change nothing, thanks to the following key lemma:

\begin{lemma}\label{inv}
Let $G$ be a connected UIG. Invariant \ref{maininv} is maintained during execution of the \texttt{minXFS} staring at an anchor.
\end{lemma}
\begin{proof}
According to Proposition~\ref{propAnchor}, the vertex with minimum degree of the last level of  BFS is an anchor. Clearly
\texttt{MaxBFS} ends at such a vertex. 
Therefore  $a=\sigma(1)$ is an anchor. Since an anchor is a leftmost vertex in some interval representation, it is  obviously simplicial. Therefore Invariant\ref{maininv} holds at step $i=1$. Let us suppose it holds at step $i$. Let $a=\sigma(i)$ and $b=\sigma(i+1)$ the vertices visited at step $i$ and $i+1$.  Notice that thanks to Invariant \ref{maininv} at step $i$, $a$ and $b$ are adjacent. Yes, a consequence of this invariant is the \texttt{minXFS}  tree of an UIG, especially  \texttt{minDFS} and  \texttt{minBFS} trees, are Hamiltonian paths. 
Suppose that  Invariant\ref{maininv} does not hold at step $i+1$. Then $Eli(i+1)\cup \{b\}$ is not a clique : it contains two nonadjacent vertices $u$ and $v$. Since the invariant holds at step $i$, at most one of them is in $Eli(i)$.

Case 1: both $u$ and $v$ belong to  $Eli(i+1)\setminus Eli(i)$. Then $\{b,a,u,v\}$ is a claw (with universal vertex $b$) and therefore $G$ is not UIG.

Case 2: one of them, say $u$ w.l.o.g, belong to  $Eli(i+1)\setminus Eli(i)$, while $v\in   Eli(i+1)\cap Eli(i)$. Since both $b$ and $v$ where in $ Eli(i)$, how could it be possible that  $b$ and not $v$ is visited at Step $i+1$?  There are 2 cases:

Case 2.1: there exists $j<i$ such that $b\in Eli(j)$ while  $v\notin Eli(j)$. We take the smallest $j$ is many exist. Since  \texttt{minXFS} is a traversal, then is must visit $b$ before $v$. Let $p=\sigma(j)$. Then $\{b,p,u,v\}$ is a claw with universal vertex $b$, and therefore $G$ is not UIG

Case 2.2: otherwise, $b$ and $u$ have exactly the same visited neighbors. Thanks to the tie-break rule of  \texttt{minXFS}, since $b$ and not $v$ was elected at step $i+1$, we have $degree(b)\le degree(v)$. But $b$ has a private unvisited neighbor $u$. Therefore $v$ must have a private neighbor $w$, that is ineligible at step $i$ (since $b$ and $v$ have the same visited neighborhood and since Invariant holds at step $i$, forbidding $a$ and $w$ to be adjacent).

Case 2.2.1 Either $u$ and $w$ are adjacent, and $\{b,u,w,v\}$ form a $C_4$, 

Case 2.2.2 Or $u$ and $w$ are not adjacent. Let $A$ be the set of the neighbors of both $b$ and $v$ visited before $b$. It is nonempty since $a\in A$. For all $x\in A$  $\{x,b,u,w,v\}$ form a Bull. The nose $x$ of this bull must be placed in the middle of these five vertices in any realizer : $x$ can not be an anchor. Let $p$ be the leftmost (smallest according to $\sigma$) vertex of $A$.  Since $p$ is not a anchor it has a predecessor $q$, visited just before it and adjacent, and not adjacent to $b$ and $v$ (and thus nor to $u$ nor to $w$).    $\{p,q,b,u,w,v\}$ is a Net and thus $G$ is not UIG. \qed 
\end{proof}

\begin{lemma}\label{indifference}
$G$ is an Unit Interval Graph iff $\sigma$ is a Indifference Ordering.
\end{lemma}
\begin{proof}
Thanks to Theorem~\ref{caracordres}, if $G$ is not an UIG then it has no Indifference Ordering. Conversely let us suppose $G$ is an UIG.
Let us suppose the ordering $\sigma$ is not Indifference. Then there exists three vertices $,c,u,v$ inducing a $P_3$, with center $c$ adjacent to the aisles $u$ and $v$, and 

Case 1: $c$ occurs before $u$ and $v$ in $\sigma$. At step $i$ such that $c=\sigma(i)$, Invariant~\ref{maininv} is violated since $u$ and $v$ are nonadjacent. So this case is impossible.

Case 2: $c$ occurs after $u$ and $v$ in $\sigma$. 
W.l.o.g assume $u <_\sigma v<_\sigma c$ 
Then according to Proposition~\ref{fourpoint},  there exists $p<_\sigma u$ and $pv\in E(G)$. At the step $j<i$ where $p$ was visited, i.e. $p=\sigma(j)$, both $v$ and $u$ (since it is visited before $v$) were eligible but $Eli(j)\cup \{p\}$ was not a clique, contradicting Invariant~\ref{maininv}. \qed
\end{proof}

\begin{theorem}\label{BFSintervalpropre}
$G$ is an Unit Interval Graph iff $\sigma$ is a Left-Right Ordering.
\end{theorem}

\begin{proof}
Lemma \ref{indifference} shows that $\sigma$ is an Indifference Ordering. But using Invariant 1, proved in Lemma \ref{inv} clearly the third pattern of cocomparability is also avoided by $\sigma$. \qed
\end{proof}
\begin{corollary}
Algorithm \ref{alg:IUG} is a linear time certifying recognition algorithm.
\end{corollary}
\begin{proof}
In case of success, it takes linear time to check if the ordering $\sigma$  is bisimplicial as can be seen in \cite{DFHP21}.
As noticed in the proof of Invariant \ref{maininv}, in case of failure, $i.e$ if $Eli(i)\cup \{\sigma(i)\}$ is not a clique, an obstruction may be found.
More precisely, in Case~1 and in Case~2.1 of the proof of Lemma~\ref{inv}, a Claw may be exhibited. In Case~2.2.1, a $C_4$. In case 2.2.2, we have a Bull whose nose is $x$.  Either the nose of the bull is the first vertex, and a 3-Sun may be exhibited since a vertex universal to all vertices but $x$ exists,  or not, and then the vertex $p$ before $x$, together with the Bull, form a Net. These four obstruction  can  be obtained in linear time.

\end{proof}
\begin{corollary}
A \texttt{minBFS} ordering obtained on an UIG is a simplicial elimination ordering.
\end{corollary}

Of course this statement is false for usual chordal graphs.

\subsection{Ordering Uniqueness}\label{unicity}

A \textbf{prime} UIG is a connected UIG without true twins, but the converse is false: think of a $P_4$ with an extra universal vertex. It is connected and has a 
non-trivial module the $P_4$. But modules of interval graphs have some special features.
In this section we will use 2 trees representations of graphs whose leaves are the vertices of the graph, namely the Modular Decomposition Tree (MDT for short) with 3 types of nodes Series, Parallel and Prime. We will use also the well known PQ-tree as defined in \cite{BoothL76} for interval graphs which represents all their intervals representations.One can find a complete analysis of the relationship between these two trees in \cite{Cresp07} for interval graphs.

\begin{theorem} \cite{HSU92}\label{Hsu}
Let  $M$ be a module of a connected  $K_{2,2}$-free graph $G$ then either $M$ is a clique or the neighbourhood of $M$ in $G$ is a clique.
\end{theorem}

When applied on chordal graph the second case of this theorem says that if $M$ is reduced to a single vertex $m$ in $G$, then $m$ is simplicial.
But we can prove a little more for UIG.

\begin{definition}
A \emph{balanced complement of a bipartite chain graph} is a graph in which the vertices can be partitioned into 2 cliques $A$ and $B$ of the same size with no 2 independent edges between  them.

\end{definition}

See Figure \ref{2-chainG} for an example of bipartite chain graph and its complement.

\begin{lemma}\label{chain-graph}
Let $G$ be a connected IUG without twins, then every non trivial prime module $M$ of $G$, $G[M]$ is a balanced complement of a bipartite chain graph. Furthermore $G$ is reduced to  $M$ + an universal vertex and it has a unique representation (up to mirror).
\end{lemma}

\begin{proof}
Notice that as $M$ is prime, it is well-known that $G[M]$ and $\overline{G[M]}$ are connected.
Since $M$ is non trivial and $G$ connected, using Theorem \ref{Hsu}, $M$ is adjacent to a clique $C$. Suppose that some vertex $v \in C$ has a neighbour outside $M$ and $C$ then we can find a $K_{1,3}$.

Therefore $C$ is a set of twins and using our hypothesis it must be  reduced to a single vertex $u$ universal to $M$. So $V(G)=M \cup \{u\}$.

Since $G$ is an UIG, let us consider a unit interval representation of $G$.
Let us fix $I(u)=[l_u, r_u]$.  For a given $m \in M$, $I(m)=[l_m, r_m]$
at least one of $l_m$ and $r_m$ belong to $I(u)$ because $u$ is universal.
But it is not possible that both $l_m$ and $r_m$ belong to $I(u)$ because of the unit length. So we can partition the vertices of $M$ into the lefties ($L$) and the righties  ($R$) with respect to $I(u)$.
Since there is at least one non-edge in $G[M]$, then $L \neq \emptyset$ and $R \neq \emptyset$.

$R, L$ are therefore cliques in $G[M]$ which is then  a cobipartite graph. Furthermore 2 independent edges between $L, R$ generate a 4 four cycle in $G[M]$ a contradiction for an interval graph. So $G[M]$ is a balanced complement of a 2-chain graph.

Let $x \in L$ and let $f(x)$ be the vertex in $R$ such as $I(f(x))$ is the closest to the right of $I(x)$
and non intersecting it.  Since $x$ is not universal such an $f(x)$ exists and is unique.
Since there are no twins for every $x,x' \in L$. $f(x) \neq f(x')$. So $|L| \leq |R|$.
Using symmetry, $|L|=|R|$ and $G[M]$ is balanced.  

$M$ as a prime IUG has a unique representation (up to mirror), furthermore the universal vertex $u$ has also a unique position inside a unit length interval representation between the sets $L$ and $R$.
In other words, the PQ-tree associated with $G$ has a $Q$-node with ordered child $L , u, R$.

\qed
\end{proof}

Of course in a connected $K_{1,3}$-free graph, a vertex can only have one false twin unless a $K_{1,3}$ is created, but we can go further for the  modular decomposition of $(K_{2,2}, K_{1,3})$-free graphs.

\begin{lemma}\label{K1,3}
Let $G$ be a $(K_{2,2}, K_{1,3})$-free G, then it cannot have two consecutive prime nodes in its modular decomposition tree.
\end{lemma}
\begin{proof}
Let us consider 2 such graphs $H, H'$ and $x \in V(H)$.
We can define $G= H_{x}^{H'}$ by substituting $x$ by  $H'$ in $H$ and we will show that $G$ is no a $K_{1,3}$-free graph.

It is well-know that every prime graph contains a $P_4$ as an induced subgraph. Say that $H$
contains $a-b-c-d$ and $H"$ contains $1-2-3-4$.
First cases: $x \in \{a,b,c,d\}$.

$x=a$ then $\{b,1,4,c\}$ is a  $K_{1,3}$.

$x=b$ its neighbourhood is not clique, contradiction.

But it could the case that no $P_4$ contains $x$ in $H$. But using a result in \cite{CournierI98} we know that necessarily $x$ is unique and the head of a bull $\{x, a,b, c, d\}$.

But then $\{a,1,4,c\}$ is a  $K_{1,3}$.
\qed
\end{proof}

All the previous properties lead  to the following structural result.

\begin{theorem}\label{structural}
The modular decomposition tree of an IUG has depth $\leq 4$.
\end{theorem}

\begin{proof}
Let $G$ be an UIG. $G$ can be disconnected and the root of its MDT be a Parallel node.

\begin{claim}
If a non trivial module has a  non edge then its neighbourhood in its connected component is a clique and is exactly the  remaining of the connected component.
\end{claim}

\begin{proof}
We already notice this property in Lemma \ref{K1,3} for prime and the same proof generalizes it to modules made up with 2 false twins.
\qed

\end{proof}

Using this claim, one can see that trivially a vertex $x$ with $degree(x)>0$ can have at most one false twin $y$ and if so its connected component is a clique in which the edge $xy$ has been removed.

So the modular decomposition tree  may start with a parallel node 
but then each of its child nodes are connected and using Lemma \ref{chain-graph} either with have a prime node or a series node with 2 children a cobipartite chain-graph on  one side and a clique on the other side, the third type of connected component is just a clique having eventually 2 false twins. Using Lemma \ref{K1,3} a prime node can only be decomposed further using set of true twins.

Therefore the height of  $MDT(G)$ is at most 4. See Figure \ref{MDT} for an example of a modular decomposition tree with prime, series and parallel (denoted by \/\/) nodes  and its associated $PQ$-tree.

\qed
\end{proof}

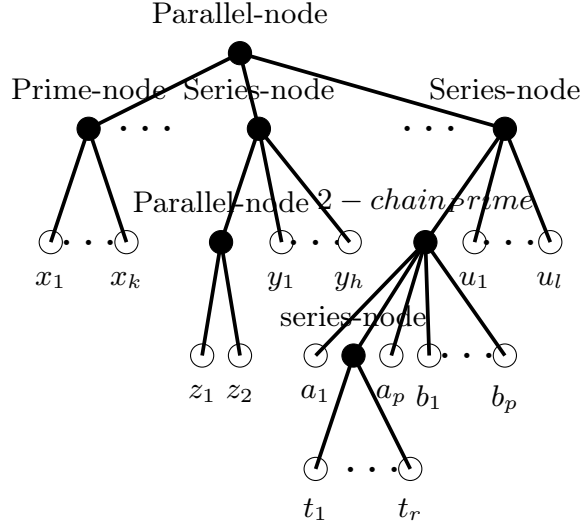
\begin{figure}
\begin{tikzpicture}[scale=0.5]

% SLICES / LARGE SETS
\coordinate (p) at (4,8) {};
\coordinate (s1) at (0,6) {};
\coordinate (s2) at (1.5,6) {};
\coordinate (s3) at (3,6) {};
\coordinate (s4) at (4.5,6) {};
\coordinate (s5) at (6,6) {};
\coordinate (s6) at (7.5,6) {};
\coordinate (s7) at (9,6) {};
\coordinate (s8) at (11,6) {};

\coordinate (s11) at (-1,3) {};

\coordinate (s115) at (0,3) {};
\coordinate (s12) at (1,3) {};
\coordinate (s31) at (1.5,3) {};
\coordinate (s32) at (3.3,3) {};
\coordinate (s41) at (3.5,3) {};
\coordinate(s415) at (6,3){};
\coordinate (s42) at (5.1,3) {};
\coordinate (s51) at (5.3,3) {};
\coordinate (s52) at (6.9,3) {};
\coordinate (s61) at (7.1,3) {};
\coordinate (s62) at (8.9,3) {};
\coordinate (s81) at (10.2,3) {};
\coordinate (s815) at (11.1,3) {};
\coordinate (s82) at (12.2,3) {};

\coordinate (s100) at (3,0) {};
\coordinate (s101) at (4,0) {};
\coordinate (s102) at (6,0) {};
\coordinate (s103) at (7,0) {};
\coordinate (s104) at (8,0) {};

\coordinate (s105) at (9,0) {};
\coordinate (s155) at (10,0) {};
\coordinate (s106) at (11,0) {};

\coordinate (s110) at (6,-3) {};
\coordinate (s1111) at (7.5,-3) {};
\coordinate (s111) at (8.5,-3) {};
% NODES %%%%%%%%%%%%%%%%%%%%%%%%%%%%%%%%%%%%%%%%%%%%%%%%%%%%%%%%%%%%%%%%%%

\node[draw, circle, minimum height=0.2cm, minimum width=0.2cm, fill=black] at (p) {};
\draw (p) node[above,scale=1.25,yshift = 2mm] {Parallel-node};
\node[scale = 2] at (s2) {$\ldots$};
\node[draw, circle, minimum height=0.2cm, minimum width=0.2cm, fill=black] at (s1) {};
\draw(s1) node[above,scale=1.25,yshift = 2mm]{Prime-node};
%\node[draw, circle, minimum height=0.2cm, minimum width=0.2cm, fill=black] at (s3) {};
%\draw (s3) node[above left,scale=1.25,xshift=-1mm] {$\ell$};
\node[draw, circle, minimum height=0.2cm, minimum width=0.2cm, fill=black] at (s4) {};
\draw(s4) node[above,scale=1.25,yshift = 2mm]{Series-node};
%\node[draw, circle, minimum height=0.2cm, minimum width=0.2cm, fill=black] at (s5) {};
%\node[draw, circle, minimum height=0.2cm, minimum width=0.2cm, fill=black] at (s6) {};
%\draw (s6) node[above right,scale=1.25,xshift=1mm] {$r$};
\node[scale = 2] at (s7) {$\ldots$};
\node[draw, circle, minimum height=0.2cm, minimum width=0.2cm, fill=black] at (s8) {};
\draw(s8) node[above,scale=1.25,yshift = 2mm]{Series-node};
%\foreach \I in {1,...,2} {
%   \coordinate (v\I) at (1.4+0.4*\I,3.2) {};
%   \node[draw, circle, fill=black, scale = 0.5] at (v\I) {};
%}

\node[draw, circle, minimum height=0.2cm, minimum width=0.2cm, fill=white] at (s11) {};
\draw (s11) node[below,scale=1.25,yshift = -2mm] {$x_1$};

\node[scale = 2] at (s115) {$\ldots$};

\node[draw, circle, minimum height=0.05cm, minimum width=0.05cm, fill=white] at (s12) {};
\draw (s12) node[below,scale=1.25,yshift = -2mm] {$x_k$};

\draw[line width = 1.4pt] (s1) -- (s11);
\draw[line width = 1.4pt] (s1) -- (s12);

\node[draw, circle, minimum height=0.2cm, minimum width=0.2cm, fill=black] at (s41) {};
\draw (s41) node[above,scale=1.25,yshift = 2mm] {Parallel-node};
\draw[line width = 1.4pt] (s4) -- (s41);

\node[draw, circle, minimum height=0.2cm, minimum width=0.2cm, fill=white] at (s42) {};
\draw (s42) node[below,scale=1.25,yshift = -2mm] {$y_1$};
\draw[line width = 1.4pt] (s4) -- (s42);

\node[scale = 2] at (s415) {$\ldots$};

\node[draw, circle, minimum height=0.2cm, minimum width=0.2cm, fill=white] at (s52) {};
\draw (s52) node[below,scale=1.25,yshift = -2mm] {$y_h$};
\draw[line width = 1.4pt] (s4) -- (s52);

\node[draw, circle, minimum height=0.2cm, minimum width=0.2cm, fill=black] at (s62) {};
\draw (s62) node[above,scale=1.25,yshift = 2mm] {$2-chain_Prime$};
\draw[line width = 1.4pt] (s8) -- (s62);

\node[draw, circle, minimum height=0.2cm, minimum width=0.2cm, fill=white] at (s81) {};
\draw (s81) node[below,scale=1.25,yshift = -2mm] {$u_1$};
\draw[line width = 1.4pt] (s81) -- (s8);

\node[scale = 2] at (s815) {$\ldots$};

\node[draw, circle, minimum height=0.2cm, minimum width=0.2cm, fill=white] at (s82) {};
\draw (s82) node[below,scale=1.25,yshift = -2mm] {$u_l$};
\draw[line width = 1.4pt] (s82) -- (s8);

\node[draw, circle, minimum height=0.2cm, minimum width=0.2cm, fill=white] at (s100) {};
\draw (s100) node[below,scale=1.25,yshift = -2mm] {$z_1$};
\draw[line width = 1.4pt] (s41) -- (s100);

\node[draw, circle, minimum height=0.2cm, minimum width=0.2cm, fill=white] at (s101) {};
\draw (s101) node[below,scale=1.25,yshift = -2mm] {$z_2$};
\draw[line width = 1.4pt] (s41) -- (s101);

\node[draw, circle, minimum height=0.2cm, minimum width=0.2cm, fill=white] at (s102) {};
\draw (s102) node[below,scale=1.25,yshift = -2mm] {$a_1$};
\draw[line width = 1.4pt] (s62) -- (s102);

\node[draw, circle, minimum height=0.2cm, minimum width=0.2cm, fill=black] at (s103) {};
\draw (s103) node[above,scale=1.25,yshift = 2mm] {series-node};
\draw[line width = 1.4pt] (s62) -- (s103);

\node[draw, circle, minimum height=0.2cm, minimum width=0.2cm, fill=white] at (s104) {};
\draw (s104) node[below,scale=1.25,yshift = -2mm] {$a_p$};
\draw[line width = 1.4pt] (s62) -- (s104);

\node[draw, circle, minimum height=0.2cm, minimum width=0.2cm, fill=white] at (s105) {};
\draw (s105) node[below,scale=1.25,yshift = -2mm] {$b_1$};
\draw[line width = 1.4pt] (s62) -- (s105);

\node[scale = 2] at (s155) {$\ldots$};

\node[draw, circle, minimum height=0.2cm, minimum width=0.2cm, fill=white] at (s106) {};
\draw (s106) node[below,scale=1.25,yshift = -2mm] {$b_p$};
\draw[line width = 1.4pt] (s62) -- (s106);

\node[draw, circle, minimum height=0.2cm, minimum width=0.2cm, fill=white] at (s110) {};
\draw (s110) node[below,scale=1.25,yshift = -2mm] {$t_1$};
\draw[line width = 1.4pt] (s103) -- (s110);

\node[draw, circle, minimum height=0.2cm, minimum width=0.2cm, fill=white] at (s111) {};

\node[scale = 2] at (s1111) {$\ldots$};

\draw (s111) node[below,scale=1.25,yshift = -2mm] {$t_r$};
\draw[line width = 1.4pt] (s103) -- (s111);

% LINKS %%%%%%%%%%%%%%%%%%%%%%%%%%%%%%%%%%%%%%%%%%%%%%%%%%%%%%%%%%%%%%%%%%

 \draw[line width = 1.4pt] (p) -- (s1);
\draw[line width = 1.4pt] (p) -- (s4);
 \draw[line width = 1.4pt] (p) -- (s8);

\end{tikzpicture}
\caption{The 3 types of connected components and their associated Modular Decomposition Tree  }\label{MDT}
\end{figure}

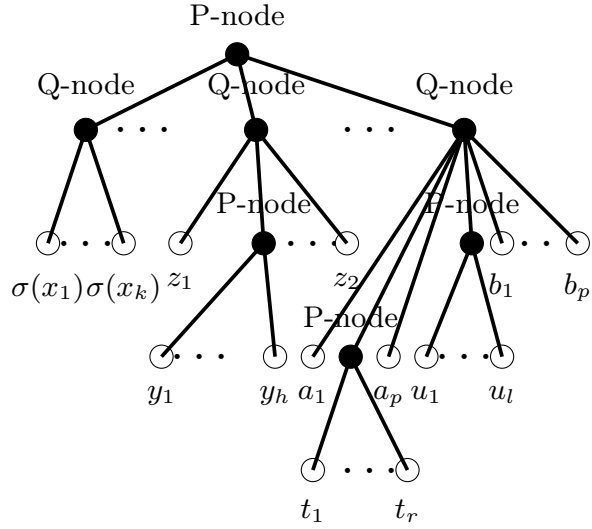
\begin{figure}
\begin{tikzpicture}[scale=0.5]

% SLICES / LARGE SETS
\coordinate (p) at (4,8) {};
\coordinate (s1) at (0,6) {};
\coordinate (s2) at (1.5,6) {};
\coordinate (s3) at (3,6) {};
\coordinate (s4) at (4.5,6) {};
\coordinate (s5) at (6,6) {};
\coordinate (s6) at (7.5,6) {};
\coordinate (s7) at (7.5,6) {};
\coordinate (s8) at (10,6) {};

\coordinate (s11) at (-1,3) {};
https://plmlatex.math.cnrs.fr/project/66ed45870663c20136f80f5d
\coordinate (s115) at (0,3) {};
\coordinate (s12) at (1,3) {};
\coordinate (s31) at (1.5,3) {};
\coordinate (s32) at (3.3,3) {};
\coordinate (s41) at (2.5,3) {};
\coordinate(s415) at (6,3){};
\coordinate (s42) at (4.7,3) {};
\coordinate (s51) at (5.3,3) {};
\coordinate (s52) at (6.9,3) {};
\coordinate (s61) at (7.1,3) {};
\coordinate (s62) at (8.9,3) {};
\coordinate (s81) at (10.2,3) {};
\coordinate (s815) at (11.6,3) {};
\coordinate (s820) at (11,3) {};
\coordinate (s82) at (13,3) {};

\coordinate (s100) at (2,0) {};
\coordinate (s1005) at (3,0) {};
\coordinate (s101) at (5,0) {};
\coordinate (s102) at (6,0) {};
\coordinate (s103) at (7,0) {};
\coordinate (s104) at (8,0) {};

\coordinate (s105) at (9,0) {};
\coordinate (s155) at (10,0) {};
\coordinate (s106) at (11,0) {};

\coordinate (s110) at (6,-3) {};
\coordinate (s1111) at (7.5,-3) {};
\coordinate (s111) at (8.5,-3) {};

% NODES %%%%%%%%%%%%%%%%%%%%%%%%%%%%%%%%%%%%%%%%%%%%%%%%%%%%%%%%%%%%%%%%%%

\node[draw, circle, minimum height=0.2cm, minimum width=0.2cm, fill=black] at (p) {};
\draw (p) node[above,scale=1.25,yshift = 2mm] {P-node};
\node[scale = 2] at (s2) {$\ldots$};
\node[draw, circle, minimum height=0.2cm, minimum width=0.2cm, fill=black] at (s1) {};
\draw(s1) node[above,scale=1.25,yshift = 2mm]{Q-node};
%\node[draw, circle, minimum height=0.2cm, minimum width=0.2cm, fill=black] at (s3) {};
%\draw (s3) node[above left,scale=1.25,xshift=-1mm] {$\ell$};
\node[draw, circle, minimum height=0.2cm, minimum width=0.2cm, fill=black] at (s4) {};
\draw(s4) node[above,scale=1.25,yshift = 2mm]{Q-node};
%\node[draw, circle, minimum height=0.2cm, minimum width=0.2cm, fill=black] at (s5) {};
%\node[draw, circle, minimum height=0.2cm, minimum width=0.2cm, fill=black] at (s6) {};
%\draw (s6) node[above right,scale=1.25,xshift=1mm] {$r$};
\node[scale = 2] at (s7) {$\ldots$};
\node[draw, circle, minimum height=0.2cm, minimum width=0.2cm, fill=black] at (s8) {};
\draw(s8) node[above,scale=1.25,yshift = 2mm]{Q-node};
%\foreach \I in {1,...,2} {
%   \coordinate (v\I) at (1.4+0.4*\I,3.2) {};
%   \node[draw, circle, fill=black, scale = 0.5] at (v\I) {};
%}

\node[draw, circle, minimum height=0.2cm, minimum width=0.2cm, fill=white] at (s11) {};
\draw (s11) node[below,scale=1.25,yshift = -2mm] {$\sigma(x_1)$};

\node[scale = 2] at (s115) {$\ldots$};

\node[draw, circle, minimum height=0.05cm, minimum width=0.05cm, fill=white] at (s12) {};
\draw (s12) node[below,scale=1.25,yshift = -2mm] {$\sigma(x_k)$};

\draw[line width = 1.4pt] (s1) -- (s11);
\draw[line width = 1.4pt] (s1) -- (s12);

\node[draw, circle, minimum height=0.2cm, minimum width=0.2cm, fill=white] at (s41) {};
\draw (s41) node[below,scale=1.25,yshift =-2mm] {$z_1$};
\draw[line width = 1.4pt] (s4) -- (s41);

\node[draw, circle, minimum height=0.2cm, minimum width=0.2cm, fill=black] at (s42) {};
\draw (s42) node[above,scale=1.25,yshift = 2mm] {P-node};
\draw[line width = 1.4pt] (s4) -- (s42);

\node[scale = 2] at (s415) {$\ldots$};

\node[draw, circle, minimum height=0.2cm, minimum width=0.2cm, fill=white] at (s52) {};
\draw (s52) node[below,scale=1.25,yshift = -2mm] {$z_2$};
\draw[line width = 1.4pt] (s4) -- (s52);

%\node[draw, circle, minimum height=0.2cm, minimum width=0.2cm, fill=black] at (s62) {};
%\draw (s62) node[above,scale=1.25,yshift = 2mm] {Q-node};
%\draw[line width = 1.4pt] (s8) -- (s62);

\node[draw, circle, minimum height=0.2cm, minimum width=0.2cm, fill=black] at (s81) {};
\draw (s81) node[above,scale=1.25,yshift = 2mm] {P-node};
\draw[line width = 1.4pt] (s81) -- (s8);

\node[scale = 2] at (s815) {$\ldots$};

\node[draw, circle, minimum height=0.2cm, minimum width=0.2cm, fill=white] at (s82) {};
\draw (s82) node[below,scale=1.25,yshift = -2mm] {$b_p$};
\draw[line width = 1.4pt] (s82) -- (s8);

\node[draw, circle, minimum height=0.2cm, minimum width=0.2cm, fill=white] at (s820) {};
\draw (s820) node[below,scale=1.25,yshift = -2mm] {$b_1$};
\draw[line width = 1.4pt] (s820) -- (s8);

\node[draw, circle, minimum height=0.2cm, minimum width=0.2cm, fill=white] at (s100) {};
\draw (s100) node[below,scale=1.25,yshift = -2mm] {$y_1$};
\draw[line width = 1.4pt] (s42) -- (s100);

\node[scale = 2] at (s1005) {$\ldots$};
\node[draw, circle, minimum height=0.2cm, minimum width=0.2cm, fill=white] at (s101) {};
\draw (s101) node[below,scale=1.25,yshift = -2mm] {$y_h$};
\draw[line width = 1.4pt] (s42) -- (s101);

\node[draw, circle, minimum height=0.2cm, minimum width=0.2cm, fill=white] at (s102) {};
\draw (s102) node[below,scale=1.25,yshift = -2mm] {$a_1$};
\draw[line width = 1.4pt] (s8) -- (s102);

\node[draw, circle, minimum height=0.2cm, minimum width=0.2cm, fill=black] at (s103) {};
\draw (s103) node[above,scale=1.25,yshift = 2mm] {P-node};
\draw[line width = 1.4pt] (s8) -- (s103);

\node[draw, circle, minimum height=0.2cm, minimum width=0.2cm, fill=white] at (s104) {};
\draw (s104) node[below,scale=1.25,yshift = -2mm] {$a_p$};
\draw[line width = 1.4pt] (s8) -- (s104);

\node[draw, circle, minimum height=0.2cm, minimum width=0.2cm, fill=white] at (s105) {};
\draw (s105) node[below,scale=1.25,yshift = -2mm] {$u_1$};
\draw[line width = 1.4pt] (s81) -- (s105);

\node[scale = 2] at (s155) {$\ldots$};

\node[draw, circle, minimum height=0.2cm, minimum width=0.2cm, fill=white] at (s106) {};
\draw (s106) node[below,scale=1.25,yshift = -2mm] {$u_l$};
\draw[line width = 1.4pt] (s81) -- (s106);

\node[draw, circle, minimum height=0.2cm, minimum width=0.2cm, fill=white] at (s110) {};
\draw (s110) node[below,scale=1.25,yshift = -2mm] {$t_1$};
\draw[line width = 1.4pt] (s103) -- (s110);

\node[draw, circle, minimum height=0.2cm, minimum width=0.2cm, fill=white] at (s111) {};

\node[scale = 2] at (s1111) {$\ldots$};

\draw (s111) node[below,scale=1.25,yshift = -2mm] {$t_r$};
\draw[line width = 1.4pt] (s103) -- (s111);

% LINKS %%%%%%%%%%%%%%%%%%%%%%%%%%%%%%%%%%%%%%%%%%%%%%%%%%%%%%%%%%%%%%%%%%

 \draw[line width = 1.4pt] (p) -- (s1);
\draw[line width = 1.4pt] (p) -- (s4);
 \draw[line width = 1.4pt] (p) -- (s8);

\end{tikzpicture}
\caption{The associated PQ-tree. Notice the differences with the Modular Decomposition Tree as for example that $\sigma$ must be an Left  Ordering of the underlying prime graph UIG.  }\label{PQT}
\end{figure}

As a consequence:

\begin{proposition}\label{freedom}
The interval representation of an IUG is essentially unique; the only degrees of freedom comes from the true twin vertices.
\end{proposition}

The above structural result applies only for non trivial module in a connected UIG, and not to a prime connected  UIG itself, but it is interesting to notice that
in \cite{HEGGERNES09} some structural results are presented using a generalization of the complement of bipartite chain graphs, namely the \textbf{2-dimensional bubbles structures}. But we will not use this characterization in the following algorithms.

\begin{lemma}\label{2=3}
For a connected graph $G$ a Bisimplicial Elimination Ordering is also a cocomp ordering.
\end{lemma}

\begin{proof}
Let $\tau$ be Bisimplicial Elimination Ordering of graph $G$. If $\tau$ is not a cocomp ordering there must exist at least one triple $(a, c, b)$ with $ab \in E(G)$ and $ac, bc \notin E(G)$ and $a <_{\tau} c <_{\tau} b$. Let us take a lexicographic smallest such triple. Since the graph $G$ is supposed to be connected, let us consider a shortest path $\mu$ from $c$ to $\{a, b\}$.

Using symmetry let us say that $\mu=[c, x_1, \dots, x_k, a]$, with $k\geq 1$. Let us now prove by induction that $x_i$ cannot be adjacent to $a$ or $b$ and $c <_{\tau} x_i$ which leads to an immediate contradiction with the existence of $\mu$.

First for i=1. $x_1b \notin E(G)$ else $\mu$ would not be a shortest path. If $x_1a \in E(G)$  If $x_1 <_{\tau} a$, the triple $(x_1, a, c)$ is a forbidden pattern. Else $a <_{\tau} x_1$ but then $(a, x_1, b)$ is a forbidden pattern.
Since $(a, x_1, b)$ is a cocomp pattern our choice of $(a,c,b)$ forces $c<_{\tau} x_1$.

Let us suppose that it is true for $x_i$ and consider $x_{i+1}$.

If $x_{i+1}b \notin E(G)$,  else $\mu$ would not be a shortest path.
Suppose $x_{i+1}a \in E(G)$. The only possible case if $x_{i+1} <_{\tau} a$, because of bisimpliciality. But then the triple $(x_{i+1}, a, x_i)$ is not simplicial.
Since $(a, x_{i+1}, b)$ is a cocomp pattern our choice of $(a,c,b)$ forces $c<_{\tau} x_1$.

Therefore such a path $\mu$ cannot exist and we are done.
\qed

\end{proof}

\begin{theorem}\label{unique}
Let $G$ be a prime  UIG.  The following orderings are the same, up to reversal : 
  \begin{itemize}
\item the Interval ordering of a realizer of $G$,
\item a Consecutive Neighborhood Ordering of $G$,
\item a Left Ordering of $G$,
\item $\texttt{minBFS}(r)$ where $r$ is left anchor of $G$. 
\item  an Indifference Ordering of $G$,
\item  a Bisimplicial Elimination Ordering of $G$, 
  \end{itemize}

\end{theorem}

\begin{proof}
First a prime graph by definition  is necessarily connected, and it is well known from \cite{SF1970} that a connected prime cocomparability graph have a unique cocomp ordering, up to reversal.
Therefore  since using Lemma \ref{2=3} all the orderings yielded by the six items are cocomp orderings we are done.
\qed
\end{proof}

Using our knowledge about modular decomposition of IUG, see Proposition  \ref{freedom}, we can push a little further.

\begin{theorem}\label{uniqueness}
Let $G$ be a connected IUG, The following orderings are the same, up to reversal and permutation of the true twin vertices: 
  \begin{itemize}
\item the Interval ordering of a realizer of $G$,
\item a Consecutive Neighborhood Ordering of $G$,
\item a Left Ordering of $G$,
\item $\texttt{minBFS}(r)$ where $r$ is left anchor of $G$. 
\item  an Indifference Ordering of $G$,
\item  a Bisimplicial Elimination Ordering of $G$, 
  \end{itemize}
\end{theorem}
\begin{proof}
Now we suppose that $G$ admits a non trivial module $M$.

If $M$ is made up with false twins, necessarily  exactly 2 say $z_1, z_2$ and their neighbour is a clique $C$.
A Consecutive Neighborhood Ordering of $G$, can only be $z_1, C, z_2$ or its reversal.

If $M$ is a prime module, then using Lemma \ref{chain-graph}  $G[M]$ is necessarily a balanced complement of a 2-chain graph in series with an universal clique see Figures \ref{MDT}, \ref{PQT}. This part is controlled with a Q-node and therefore has only one interval representation up to reversal.

If $M$ is a set of true twins, i.e. a clique, in this case the Consecutive Neighborhood Ordering property forces the twins to be contiguous. Therefore the ordering is unique up to some permutation of the true twins. \qed
\end{proof}

\textbf{Remarks:}
  If a UIG is not connected, a realizer may be built by arbitrarily numbering the connected components, then putting the realizers of the first component, then a gap (no interval), then a translation of the realizer of the second component, and so on.
In reality, one can give an explicit formula of all representations of an IUG depending on its true twins. Going to complement this can also be applied to enumerate all linear extensions of a proper interval order.
  If a UIG contains true twins $x$ and $x'$ then there is a realizer in which $Interval(x)=Interval(x')$.
 For the particular case of complement of bipartite chain graphs as seen in Section \ref{aperitif} only a unique \texttt{minBFS} is necessary for the recognition, since  a  vertex with minimal degree is always an anchor.

%%%%%%%%%%%%%%%%%%%%%%%%%%%%%%%%%%%%%%%%%%%%%%%%%%
%%%%%%%%%%%%%%%%%%%%%%%%%%%%%%%%%%%%%%%%%%%%%%%%%%
\section{Conclusion and Perspectives}
In this paper we show how to use simple BFS or DFS preprocessed with vertex degrees to recognize Split graphs, Bipartite chain graph, Trivially Perfect Graphs and Proper Interval graphs.
It could be possible to generalize Theorem \ref{BFSintervalpropre} can be generalized  to   interval graphs.
Perhaps our approach based on simple graph searches can be generalized to other well structured graph classes such as for example  $P_4$-indifference graphs, see \cite{LO1993} or \cite{HabibPV01}, or used as preprocessing in some multi-sweep algorithms.

Playing with characterizations of graph classes in terms of existing ordering of the vertices avoiding some patterns allows to partition the known characterizations of proper interval graphs between those that need forbidding 3 patterns and those only 2. But in Theorems \ref{unique} we show that non only they characterize the same class of graphs, but also when the graph is connected the several vertex orderings are the same, up to mirroring and permutation of twins. Maybe this kind of counter-intuitive phenomenon can be found elsewhere.
It should also be noticed the interesting work initiated by F. Dragan
\cite{Bhaduri08} about vertex orderings  that select at each step a minimum degree vertex in the remaining graph. They showed that such orderings yield an easy way to compute a maximum clique for a chordal graph.  

In Section~\ref{section.uig} we proved that either \texttt{minBFS} or \texttt{minDFS}, under the common name of \texttt{minXFS}, may be used to recognize Unit Interval Graphs (UIG). But in fact any graph search where the minimum degree vertex, among eligible ones, is visited, would maintain Invariant \ref{maininv} and could be called a \texttt{minXFS} and used to recognize UIG, but we do not consider yet how to implement in linear time these more generic searches.

\noindent 
\textbf{Level Layered Search (LLS) instead of BFS?}
In many applications BFS is not really needed, a standard LLS could do the same work, as for  example to find an anchor in an IUG.
It would be interesting to investigate in which cases we really need a BFS with its queue data structure.
Not only it is easier to parallelize LLS than BFS, but also deciding whether a vertex is the last one in a BFS is NP-complete \cite{CharbitHM14}, while  polynomial for LLS.

\bibliographystyle{splncs04}
\bibliography{references.bib}

\end{document}